\documentclass[12pt]{amsart}
\usepackage{amssymb}
\usepackage{enumerate,color}
\usepackage{graphicx}
\usepackage[text={6in,9in},centering]{geometry}
\usepackage{subfigure}
\usepackage{epstopdf}

\usepackage{mathrsfs}

\theoremstyle{plain}
  \newtheorem{thm}{Theorem}[section]
  \newtheorem{cor}[thm]{Corollary}
  \newtheorem{lem}[thm]{Lemma}
  \newtheorem{prop}[thm]{Proposition}
\theoremstyle{definition}
  
  \newtheorem{exmp}{Example}
\theoremstyle{remark}
  
\theoremstyle{question}
  \newtheorem{ques}{Question}

\DeclareMathOperator\bdim{dim_B}

\def\R{\mathbb{R}}
\def\D{\mathcal{D}}
\def\A{\mathcal{A}}
\def\B{\mathcal{B}}
\def\V{\mathcal{V}}

\def\Z{\mathbb{Z}}
\def\wQ{\widetilde{Q}}
\def\C{\mathscr{C}}

\begin{document}
\title[Gap sequences and Topological properties of Bedford-McMullen sets]{Gap sequences and Topological properties of Bedford-McMullen sets}
\author{Zhen Liang}
\address{Department of Statistics,  Nanjing Audit University, Nanjing 211815, China}
\email{zhenliang@nau.edu.cn}
\author{Jun Jie Miao}
\address{School of Mathematical Sciences, East China Normal University, Shanghai 200241,  China}

\email{jjmiao@math.ecnu.edu.cn}
\thanks{The research of Miao is partially supported by Shanghai Key Laboratory of PMMP (18dz2271000).}
\author{Huo-Jun Ruan}
\address{School of Mathematical Sciences, Zhejiang University, Hangzhou 310027, China}
\email{ruanhj@zju.edu.cn}
\thanks{The research of Ruan is partially supported by NSFC grant 11771391, ZJNSF grant LY22A010023 and the Fundamental Research Funds
for the Central Universities of China grant 2021FZZX001-01.}



\begin{abstract}
In this paper, we study the topological properties and the gap sequences of Bedford-McMullen sets. First, we introduce a topological condition, the component separation condition (CSC), and a geometric condition, the  exponential rate condition (ERC). Then we prove that the CSC implies the ERC, and that both of them are sufficient conditions for obtaining the asymptotic estimate of gap sequences. We also  explore topological properties of Bedford-McMullen sets and prove that all normal Bedford-McMullen sets with infinitely many connected components satisfy the CSC, from which we obtain the asymptotic estimate of the gap sequences of  Bedford-McMullen sets without any restrictions.  Finally, we apply our result to Lipschitz equivalence.
\end{abstract}


\maketitle

\section{Introduction }

\subsection{Cardinality of connected components of $\delta$-neighbourhood}
Let $A$ be a nonempty compact subset of $\R^d$. Given a real number $\delta >0$, we define $A_\delta:=\{y\in \R^d : |x-y|\leq \delta \textit{ for some } x\in A\}$ for the $\delta$-\textit{neighbourhood} (or $\delta$-\textit{fattening}) of the set $A$. Let $\mathscr{C}(B)$ denote the set of all connected components of a set $B$, and denote by $\# E$ the cardinality of a set $E$. Then $\#\mathscr{C}(A_\delta)$ is finite for all $\delta>0$ and is decreasing with respect to $\delta$.

It is clear that if $A$ is connected or contains only finitely many connected components, then there exists $\delta_0>0$, such that $\#\mathscr{C}(A_\delta)=\#\mathscr{C}(A)$ for all $0<\delta<\delta_0$. Things become quite different if $A$ contains infinitely many connected components. In this case, it is natural to study the \emph{rate} at which $\# \mathscr{C}(A_\delta)$ tends to infinity as $\delta$ goes to $0$. More precisely, it is natural to study whether there exists a constant $\gamma>0$ satisfying
\[
  \#\mathscr{C}(A_\delta) \asymp \delta^{-\gamma} \quad (\delta\to 0),
\]
where for two positive functions $f_1$ and $f_2$ on $(0,+\infty)$, we use $f_1(\delta)\asymp f_2(\delta) \; (\delta\to 0)$, or simply $f_1(\delta)\asymp f_2(\delta)$, to mean that there exist $\delta_0>0$ and $c>0$, such that $c^{-1}f_2(\delta)\leq f_1(\delta)\leq c f_2(\delta)$ for all $0<\delta<\delta_0$.

In this paper, we completely characterize the rate at which $\# \mathscr{C}(A_\delta)$ tends to infinity in the case that $A$ is a  Bedford-McMullen set.

\subsection{Definition of the gap sequence and relevant works}

It is quite natural to define gap sequences in the one-dimensional case. Let $A$ be a compact subset of $\mathbb{R}$. An open set $(a,b)$ is called a \emph{gap} of $A$ if $[a,b]\cap A=\{a,b\}$. We list all lengths of gaps of $A$ in non-increasing orders, and call this (finite or infinite) sequence \emph{the gap sequence} of $A$.  For example, the gap sequence of the classical Cantor middle-third set is
\begin{equation*}
\frac{1}{3},\ \ \frac{1}{3^2},\ \ \frac{1}{3^2}\ \dots,\ \underbrace{\frac{1}{3^k}\ \ ,\dots,\ \ \frac{1}{3^k}}_{2^{k-1}},\ \dots.
\end{equation*}

Gap sequences have been applied to explore the geometric properties and dimensions of fractals. In particular, it is often related to the upper bounds of box dimensions. For example, it  was applied to investigate the box dimension of cut-out sets by Besicovitch and Taylor~\cite{BesTa54}.  We refer readers to~\cite{Falco95,LapMa95,LapPo93,XioWu09} for more applications of gap sequences.

In 2008, Rao, Ruan and Yang \cite{RaRuY08} extended the notion of gap sequences to general Euclidean spaces. Let $A$ be a compact subset of $\mathbb{R}^d$. Given $\delta>0$, two points $x,y\in A$ are said to be $\delta$-equivalent (in $A$)  if there exists a sequence $\{x_1=x,x_2,\ldots,x_{\ell-1},x_{\ell}=y\}\subset A$ such that $|x_{i+1}-x_i|\leq \delta$ for all $1\leq i\leq \ell-1$. We denote by $h_A(\delta)$, or $h(\delta)$ for short, the number of $\delta$-equivalent classes of $A$.
Let $\{\delta_{j}\}_{j\geq1}$ be the discontinuity points (or jump points) of $h_A(\delta)$ in decreasing order. Obviously $h_A(\delta)$ equals $1$ on $[\delta_1,+\infty)$, and is constant on $[\delta_{j+1},\delta_j)$ for $j\geq 1$. Let $h_A(\delta_j^-)$ be the left limit of $h_A(\delta)$ at $\delta_j$, i.e., $h_A(\delta_j^-)=\lim_{\delta\to \delta_{j}^{-}}h_A(\delta)$.
We call $m_{j}=h_A(\delta_j^-)-h_A(\delta_{j})$ the \textit{multiplicity} of $\delta_{j}$, and define the \emph{gap sequence} of $A$, denoted by $\{g_k(A)\}$, to be the sequence:
\begin{equation}\label{eq:gap-def}
  \underbrace{\delta_{1},\dots,\delta_{1}}_{m_{1}},\underbrace{\delta_{2},\dots,
  \delta_{2}}_{m_{2}},\dots,\underbrace{\delta_{j},\dots,\delta_{j}}_{m_{j}},\dots.
\end{equation}

It is easy to check that the definition of gap sequences in \cite{RaRuY08} coincides with the natural definition in the one-dimensional case. Rao, Ruan and Yang also proved the following interesting property: if two compact subsets $E$ and $F$ of $\mathbb{R}^d$ are Lipschitz equivalent, then their gap sequences $\{g_k(E)\}$ and $\{g_k(F)\}$ are \emph{comparable}, i.e., there exists $c>0$ such  that $c^{-1}<g_k(E)/g_k(F)<c$ for all $k$. Moreover, they proved that under certain constraints, the upper box dimension of a compact set $E$ is determined by its gap sequence:
\begin{equation}\label{Bdim}
{\overline{\dim}}_{\mathrm B} E= \limsup_{k\to\infty}\frac{\log k}{-\log g_k(E)},
\end{equation}
while this result was proved by  Tricot  in the one-dimensional case~ \cite{Tricot}.

Recently, there have been several works on estimating gap sequences. Most of them focus on fractal sets with the totally disconnected condition. Deng, Wang and Xi \cite{DeWaX15} studied the gap sequences of $C^{1+\alpha}$ self-conformal sets satisfying the strong separation condition, and they proved that \eqref{Bdim} holds for box dimension. Deng and Xi \cite{DeXi16} studied the gap sequences of graph-directed sets satisfying the strong separation condition. Rao, Ruan and Yang studied the gap sequences of totally disconnected Bedford-McMullen sets with certain constraints, see Example~3.3 in \cite{RaRuY08}. Miao, Xi and Xiong \cite{MXX} studied the gap sequence of a totally disconnected Bedford-McMullen set $E$ without any constraints, and showed that
\[
  g_k(E)\asymp k^{-{1/d}},
\]
where $d$ depends on how many empty rows there are in the geometric construction of $E$. See section~2 for the definition of empty row. This provides counterexamples to~\eqref{Bdim} in the general self-affine setting.
Here for two positive functions $f_1$ and $f_2$ on $\mathbb{Z}^+$, we use $f_1(k)\asymp f_2(k)$ to mean that there exists
$c>0$, such that $c^{-1}f_2(k)\leq f_1(k)\leq cf_2(k)$ for all $k\in \mathbb{Z}^+$.

More recently, Liang and Ruan~\cite{LiRu19} studied the gap sequences of fractal squares. In their setting, all possible nontrivial connected components of fractal squares are parallel to either $x$-axis or $y$-axis.

Inspired by~\cite{LiRu19, MXX, RaRuY08}, the following question arises quite naturally: \emph{Can we characterize the asymptotic estimate of the gap sequences of Bedford-McMullen sets without any restrictions?}

From \cite{DeWaX15,MXX}, we have the following relationship between $h_E(\delta)$ and the gap sequences:
\begin{equation}\label{eq:rel-h-gap}
   h_E(\delta)\asymp \delta^{-\gamma} \quad (\delta\to 0)  \iff g_k(E) \asymp k^{-1/\gamma},
\end{equation}
where $\gamma>0$ is a constant, see Proposition~\ref{prop1} for details. Thus, in order to obtain the asymptotic behaviour of gap sequences, it suffices to know the asymptotic behaviour of $h_E(\delta)$. Using the techniques in \cite{MXX} and also some ideas in \cite{Xiao}, we explore local topological structure and metric structure of Bedford-McMullen sets. As a result, we finally obtain the asymptotic estimate of $h_E(\delta)$ and thus that of $\{g_k(E)\}$. Noticing that by definition, $\#\mathscr{C}(A_{\delta/2})=h_A(\delta)$, we also obtain the asymptotic estimate of $\#\mathscr{C}(A_{\delta})$.

Bedford-McMullen sets are classical self-affine sets. They often serve as a testing ground for questions, conjectures or counterexamples. Originally, these sets were studied independently by  Bedford~\cite{Bedfo84} and McMullen~\cite{McMul84}.  There are many interesting works to study the dimensions and Lipschitz equivalence of these sets and their varieties. We refer the reader to \cite{Baran07, FFK, Fraser20,FraOl11,KenPe96,King95,LalGa92,LiLiMi13,Olsen11,Peres94, RaYaZh20} and references therein.

The paper is organized as follows. In section~2, we recall the definition of Bedford-McMullen sets and present the main result. We also introduce two conditions: the component separation condition (CSC for short) and the exponential rate condition (ERC for short). In section~3, we prove that both the CSC and the ERC imply the certain asymptotic behaviour of gap sequences.  The main result is proved in section~4. In section~5, we apply our main result to Lipschitz equivalence of Bedford-McMullen sets.  Finally, we give further remarks in  section~6.

\section{Definitions, notations and main results}
\setcounter{equation}{0}

\subsection{Definitions and notation}\label{sec_BM}
Fix integers $m$, $n$ and $N$ such that $n\geq m\geq 2$  and $1< N\leq mn$.
Let $\mathcal{D}$ be a subset of $\{0,\dots,n-1\}\times
\{0,\dots,m-1\}$ with $\#\D=N$. Write
$$
M=\#\{j\colon (i,j)\in \mathcal{D} \textrm{ for some } i\}.
$$
For each $w\in \mathcal{D}$, we define an
affine transformation $\Psi_w$ on~$\R^2$ by
\begin{equation}\label{eq:Sk}
  \Psi_w(x)=T(x+w),
\end{equation}
where $T=\operatorname{diag}(n^{-1},m^{-1})$. Then the family
$\{\Psi_w\}_{w\in\mathcal{D}}$ forms a self-affine \emph{iterated function system} (IFS), see Figure~\ref{McMfig}. According
to Hutchinson~\cite{Hutch81}, there exists a unique nonempty compact subset  $E=E(n,m,\D)$ of $\R^2$,  called
a \emph{Bedford-McMullen set} or  a \emph{Bedford-McMullen carpet}  \cite{Bedfo84,McMul84}, such that
$$
  E=\bigcup_{w\in \mathcal{D}}\Psi_w(E).
$$
In the special case where $n=m$, the Bedford-McMullen set $E$ is a self-similar set, which we term a \emph{fractal square}, or a \emph{generalized Sierpinski carpet}, denoted  by $F(n,\D)$. See~\cite{LaLuRa13,LiRu19,Roin10, RuWa17, XiXi10} for various studies on fractal squares.

  \begin{figure}[tbhp]
  \subfigure{
    \begin{minipage}[t]{0.3\linewidth}
      \begin{center}
           {\includegraphics[height=4.0cm]{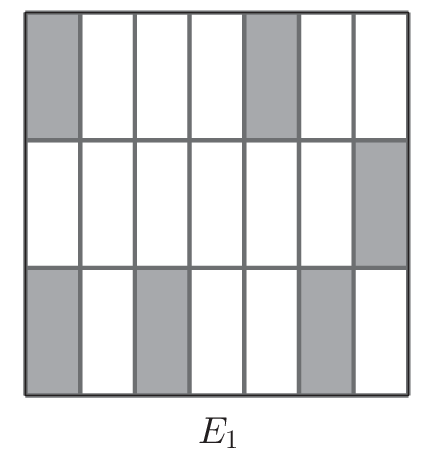}}
      \end{center}
  \end{minipage}
  }
  \subfigure{
    \begin{minipage}[t]{0.3\linewidth}
       \begin{center}
           {\includegraphics[height=4.0cm]{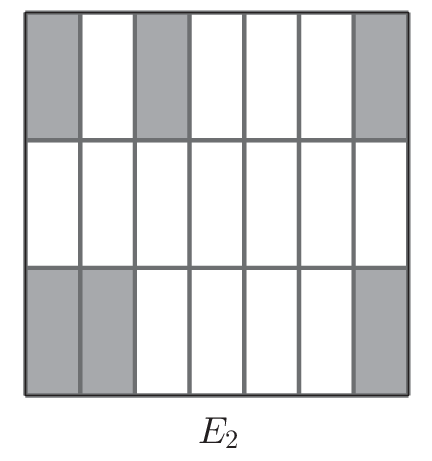}}
       \end{center}
    \end{minipage}
  }
  \subfigure{
    \begin{minipage}[t]{0.3\linewidth}
       \begin{center}
           {\includegraphics[height=4.0cm]{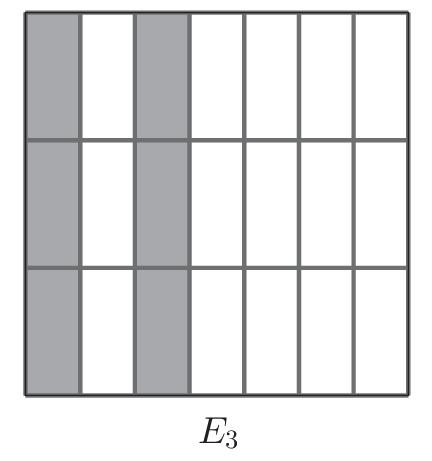}}
       \end{center}
    \end{minipage}
  }
  \caption{{$M=3$ in $E_1$ and $E_3$, $M=2$ in $E_2$, where $n=7$, $m=3$, $N=6$.}}
  \label{McMfig}
 \end{figure}

As stated in \cite{Bedfo84,McMul84}, the box dimension of the
Bedford-McMullen set~$E$ is
\begin{equation}\label{eq:DboxFormula}
  \bdim E=\frac{\log N-\log M}{\log n}+ \frac{\log M}{\log m}  .
\end{equation}
In the case that $n=m$, the dimension formula simplifies to
\[
\bdim E=\frac{\log N}{\log n},
\]
which gives the box dimension formula for fractal squares.
The formula for $n>m$ indicates that the box dimension depends not only on $N$ but also on $M$. Thus,
moving a rectangle to an empty row (or creating a new empty row) will change the dimension,
which is different to self-similar construction. Here, given $j\in \{0,1,\ldots,m-1\}$, the $j$-th row is called \emph{empty} if $\#\{i: (i,j)\in \D\}=0$.


There is a special class of Bedford-McMullen sets whose geometric structure is different to others. Let $E$ be a Bedford-McMullen set. We say $E$ is \emph{simple} if one of the following properties is satisfied:
\begin{itemize}
  \item there exists $\A\subset\{0,1,\ldots,n-1\}$ such that $\mathcal{D}= \A\times\{0,1,\ldots,m-1\}$;
  \item there exists $\B\subset \{0,1,\ldots,m-1\}$ such that $\mathcal{D}= \{0,1,\ldots,n-1\}\times \B $.
\end{itemize}
We remark that in the case that $E$ is simple, it has uniform fibres, hence, its Hausdorff and box dimensions are
equal. See section~5 for the definition of uniform fibres.
We say $E$ is \emph{normal} if $E$ is not simple.
The set $E_3$ in Figure~\ref{McMfig} is simple and satisfies the first property in the definition.


\subsection{The CSC condition and the ERC condition}

Let $\Sigma^{0}=\{\emptyset\}$.
For $k=1,2, \ldots$, we write $\Sigma^{k} = \{ w_{1} w_2 \cdots
w_{k}\colon  w_{j} \in \mathcal{D}, j=1,2,\ldots,k \}$. Elements in $\Sigma^k$ are called \emph{words with length $k$}.

Let $\{\Psi_w\}_{w\in\mathcal{D}}$ be the self-affine IFS as in~\eqref{eq:Sk}.
For each $w=w_1\cdots w_k\in\Sigma^k$, we define
\[
  \Psi_w=\Psi_{w_1}\circ\cdots\circ \Psi_{w_k}.
\]
Denote the unit square $[0,1]^2$ by $Q_0$. For each $k\in \Z^+$, we write
\begin{eqnarray*}
Q_{k}=\bigcup_{w\in \Sigma^k} \Psi_w(Q_0),  \qquad\qquad \wQ_k=\bigcup_{(i,j)\in \{-1,0,1\}^{2}} \Big(Q_k+(i,j)\Big).
\end{eqnarray*}
 See Figure~\ref{QXi1} for an example.
\begin{figure}[tbhp]
  \subfigure{
    \begin{minipage}[t]{0.45\linewidth}
      \begin{center}
           {\includegraphics[height=5.0cm]{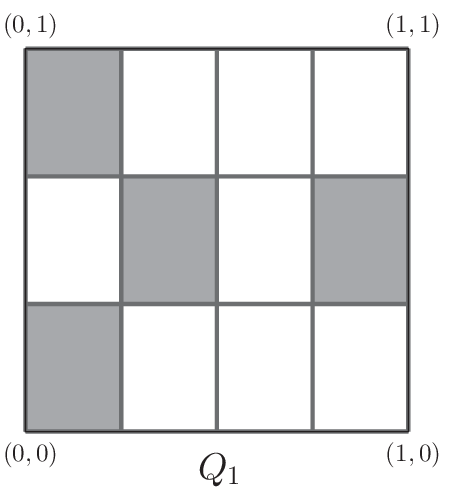}}
      \end{center}
  \end{minipage}
  }
  \subfigure{
    \begin{minipage}[t]{0.45\linewidth}
       \begin{center}
           {\includegraphics[height=5.0cm]{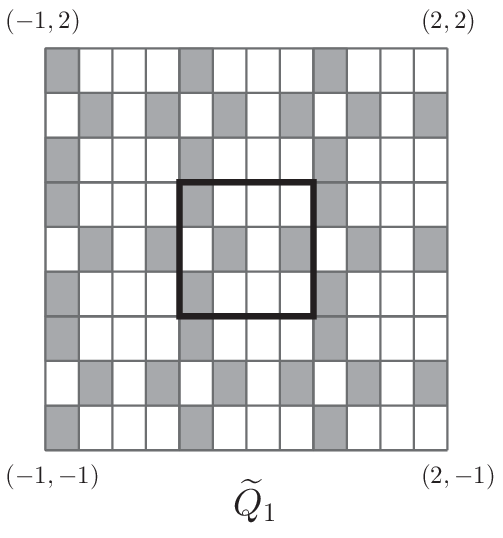}}
       \end{center}
    \end{minipage}
  }
  \caption{{$Q_1$ and $\wQ_1$ for $n=4$, $m=3$ and $N=4$.}}
  \label{QXi1}
\end{figure}

Next, we introduce two conditions for studying gap sequences. One concerns the topological structure, and the other concerns the geometric structure of fractal sets.
Both are important for the estimation of gap sequences.

We say $E$ satisfies the \emph{component separation condition} (or \emph{CSC} for short), if there exist $k_0\in \Z^+$ and a connected component $C$ of $Q_{k_0}$ such that $C$ is also a connected component of $\wQ_{k_0}$, i.e., for every connected component $C'\in \wQ_{k_0}$, either $C'=C$ or $C'\cap C=\emptyset$.

We say $E$ satisfies the \emph{exponential rate condition} (or \emph{ERC} for short) if
\begin{equation}\label{eqiv}
  \#\mathscr{C}(Q_k)\asymp N^k.
\end{equation}
Notice that $\#\C(E)\geq \#\C(Q_k)$ for all $k\geq 1$. Thus, if $E$ satisfies the ERC, then $E$ contains infinitely many connected components.

For the  Bedford-McMullen  set in Figure~\ref{QXi1}, we have $\#\C(Q_k)=2\cdot 4^{k-1}$ for all $k\geq 1$ so that it satisfies the ERC.
There exist examples which do not satisfy the ERC. For example, $\#\mathscr{C}(Q_k)=2^k$ for the set $E_3$ in Figure~\ref{McMfig}; $\#\mathscr{C}(Q_k)=2$ for  the set $E$  in Example~\ref{RX}.

\subsection{Results and examples}
It is not difficult to obtain the estimate of gap sequences of  Bedford-McMullen sets in the simple case. For the normal case, we need the following result.

\begin{thm}\label{thm1}
Let $E$ be a Bedford-McMullen set.
\begin{enumerate}

\item If $E$ satisfies the component separation condition, then $E$ satisfies the exponential rate condition.

\item If $E$ satisfies the exponential rate condition,  then the gap sequence satisfies
\begin{equation}\label{eq:gap-McMullen}
g_k(E)\asymp\begin{cases}
           k^{-1/\bdim E}, & \text{if $M<m$}; \\
           k^{-\log n/\log N}, & \text{if $M=m$}.
         \end{cases}
\end{equation}
\end{enumerate}
\end{thm}

By using this result, we can obtain our main result as follows.

\begin{thm}\label{Thm2}
Let $E$ be a Bedford-McMullen set  with infinitely many connected components.
\begin{enumerate}
  \item If $E$ is simple, then $g_k(E)\asymp k^{-1/(\bdim E-1)}.$
  \item If $E$ is normal, then $E$ satisfies the component separation condition and \eqref{eq:gap-McMullen} holds.
\end{enumerate}
\end{thm}

We remark that by \eqref{eq:rel-h-gap} and $\#\mathscr{C}(E_{\delta/2})=h_E(\delta)$, we also obtain the rate at which $\# \mathscr{C}(E_\delta)$ tends to infinity as $\delta$ goes to $0$.

By Theorem~\ref{Thm2}, we immediately obtain the following result on fractal squares,
which generalizes the main result of~\cite{LiRu19}.

\begin{cor}\label{Thm3}
Let $F(n,\D)$ and $F(n,\D')$  be two fractal squares with infinitely many connected components and $\#\D=\#\D'$. Then their gap sequences are comparable if and only if either both $F(n,\D)$ and $F(n,\D')$ are normal, or both $F(n,\D)$ and $F(n,\D')$ are simple.
\end{cor}
\begin{proof}
  From the box dimension formula,
  \[
    \bdim F(n,\D)=\bdim F(n,\D')=\frac{\log \#\D}{\log n}.
  \]
  Combining this with Theorem~\ref{Thm2}, the conclusion holds.
\end{proof}

We end this section by giving some examples to illustrate the main result.
\begin{exmp}\label{ex_1}
  Let $E_1$, $E_2$ and $E_3$ be the Bedford-McMullen sets defined in Figure~\ref{McMfig}. Clearly, the box dimensions and gap sequences vary by changing the translations of affine transformations.
\begin{eqnarray*}
\begin{array}{llc}
  \bdim E_1 =1+\frac{\log 2}{\log 7},\qquad &  g_k(E_1)\asymp \big(\frac{1}{k}\big)^{\log 7/\log6}, \qquad& M=m=3; \\
  \bdim E_2 =\frac{\log 3}{\log 7}+\frac{\log 2}{\log 3},  & g_k(E_2)\asymp \big(\frac{1}{k}\big)^{1/\bdim E_2} , & m>M=2 ; \\
  \bdim E_3 =\bdim E_1 , & g_k(E_3)\asymp \big(\frac{1}{k}\big)^{\log 7/\log 2} , & E_3 \textit{ is simple}.
\end{array}
\end{eqnarray*}
\end{exmp}

\begin{exmp}\label{RX}
 There exist Bedford-McMullen sets with finitely many connected components. Let $E$ be  the Bedford-McMullen set defined in Figure~\ref{McM3}. One can check that $E$ contains exactly $2$ connected components. This set is similar to the fractal square constructed by Xiao in Example~2.1 in~\cite{Xiao}. We remark that Xiao's  criteria for finitely many connected components of fractal squares are also applicable to Bedford-McMullen sets.
\end{exmp}

  \begin{figure}[tbhp]
  \subfigure[The initial structure]{
    \begin{minipage}[t]{0.45\linewidth}
      \begin{center}
           {\includegraphics[height=4.0cm]{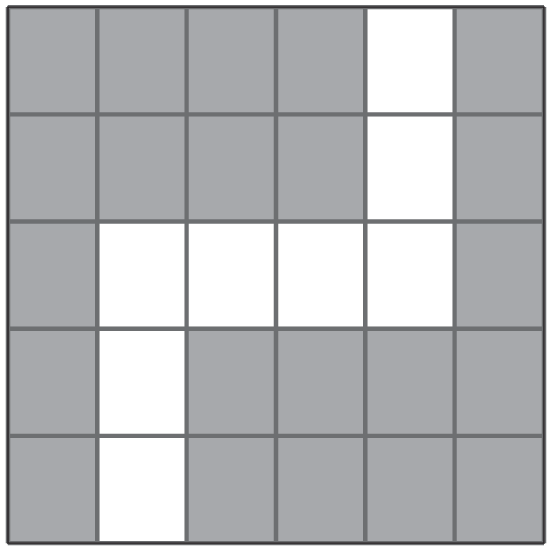}}
      \end{center}
  \end{minipage}
  }
  \subfigure[$E$]{
    \begin{minipage}[t]{0.45\linewidth}
       \begin{center}
           {\includegraphics[height=4.0cm]{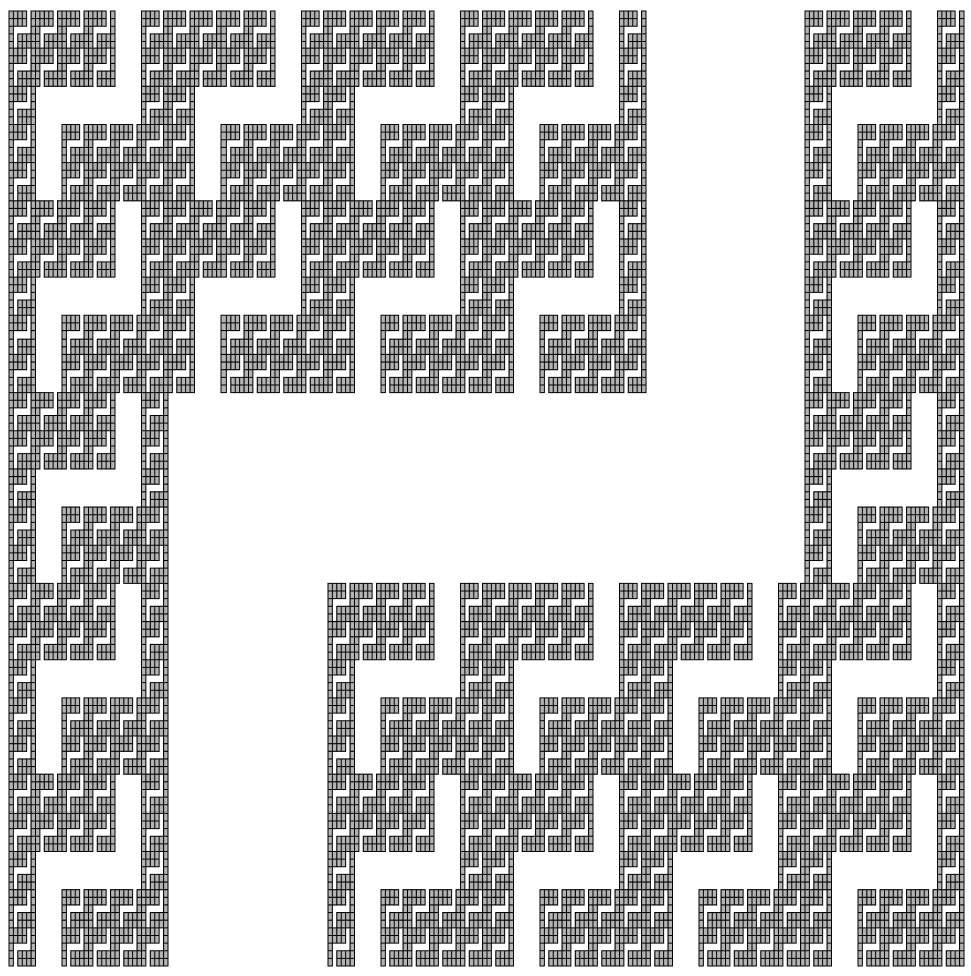}}
       \end{center}
    \end{minipage}
  }
  \caption{A Bedford-McMullen set with $2$ connected components.}
  \label{McM3}
 \end{figure}

\section{ Gap sequences of  Bedford-McMullen  sets}\label{sec3}
\setcounter{equation}{0}

In the rest of the paper, we always use \emph{component} to mean \emph{connected component} for simplicity.

First we show that the CSC implies the ERC.
\begin{lem}\label{CSC-ERC}
  Let $E$ be a Bedford-McMullen set satisfying the component separation condition. Then $\#\mathscr{C}(Q_k) \asymp N^k$.
\end{lem}
\begin{proof}
  For each $k\in\Z^+$, $Q_k$ is the union of $N^k$ rectangles so that $\#\mathscr{C}(Q_k)\leq N^k$.
  On the other hand, by the definition of the component separation condition, there exist $k_0\in \Z^+$ and a component $C$ of $Q_{k_0}$ such that $C$ is also a component of $\widetilde{Q}_{k_0}$. Thus, given a positive integer $k> k_0$, for every $v\in \Sigma^{k-k_0}$, the set $\Psi_v(C)$ is a component of $ Q_{k}$. It follows that $\#\mathscr{C}(Q_k)\geq N^{k-k_0}$.
\end{proof}

In order to prove the second part of Theorem~\ref{thm1}, we need the following fact obtained in  \cite{DeWaX15, MXX}.  We add a proof here for the convenience of readers.

\begin{prop}[\cite{DeWaX15, MXX}]\label{prop1}
Let $A$ be a compact subset of $\R^d$. Assume that $\gamma>0$. Then
\[
  g_k(A) \asymp k^{-1/\gamma} \iff  h_A(\delta)\asymp \delta^{-\gamma}  \quad (\delta\to 0).
\]
\end{prop}
\begin{proof} We remark that both $g_k(A) \asymp k^{-1/\gamma}$ and $h_A(\delta)\asymp \delta^{-\gamma} \; (\delta\to 0)$ imply that $A$ contains infinitely many components. In this proof, we denote $g_k(A)$ by $\alpha_k$ for simplicity.

  ($\Longleftarrow$). Assume that  $h_A(\delta)\asymp \delta^{-\gamma}  \; (\delta\to 0)$.  Then  it is easy to see that there exists $c>0$ such that
   \begin{equation}\label{eq:propMMX-1}
       c^{-1} \delta^{-\gamma}\leq h_A(\delta)\leq c \delta^{-\gamma}
   \end{equation}
   for all $0<\delta\leq \alpha_1$. Let $\{\delta_{j}\}_{j\geq1}$ be the discontinuity points of $h_A(\delta)$ in decreasing order.  Then $h_A(\delta_{j}^-)=h_A(\delta_{j+1})$. Thus,
   $\alpha_k=\delta_j$ for $ h_A(\delta_j)\leq k < h_A(\delta_{j}^-),$
   which implies that
   \begin{equation}\label{eq:gap-prop}
     h_A(\alpha_k) \leq k<h_A(\alpha_k^-),\quad k\geq 1.
   \end{equation}
  Combining this with \eqref{eq:propMMX-1}, we have $k\geq h_A(\alpha_k)\geq c^{-1} \alpha_k^{-\gamma}$ so that
  $\alpha_k\geq  c^{-1/\gamma}k^{-1/\gamma}$ for all integers $k>0$.

  On the other hand, given $\delta\in (0,\alpha_k)$ where $k$ is a fixed positive integer, it follows from \eqref{eq:propMMX-1} and \eqref{eq:gap-prop} that $ c\delta^{-\gamma}>h_A(\delta)\geq h_A(\alpha_k^-)>k$. Thus $\delta <c^{1/\gamma} k^{-1/\gamma}$. Letting $\delta\to \alpha_k^-$, we obtain that  $\alpha_k\leq  c^{1/\gamma}k^{-1/\gamma}$.


  ($\Longrightarrow$). Suppose now that $c^{-1} k^{-1/\gamma}\leq \alpha_k\leq c k^{-1/\gamma}$ for all integers $k>0$, where $c>0$ is a constant independent of $k$.
  For each $0<\delta<\alpha_1$, we choose an integer $k>0$ such that $\alpha_{k+1}\leq \delta<\alpha_k$. From \eqref{eq:gap-prop},
  $$
  h_A(\delta) \leq h_A(\alpha_{k+1}) \leq k+1\leq 2k\leq 2c^\gamma\alpha_k^{-\gamma}<2c^\gamma\delta^{-\gamma}
  $$
  and
  $$
  h_A(\delta)\geq h_A(\alpha_k^-)> k\geq \frac{1}{2}(k+1)\geq \frac{1}{2}c^{-\gamma}\alpha_{k+1}^{-\gamma}\geq  \frac{1}{2}c^{-\gamma}\delta^{-\gamma}.
  $$
  Thus  $h_A(\delta)\asymp \delta^{-\gamma}  \;(\delta\to 0)$  so that the conclusion holds.
\end{proof}

\medskip

Given $k\in \Z^+$, for each $w\in\Sigma^k$, we call $\Psi_w(Q_0)$ an \emph{elementary rectangle} of $Q_k$. Write
\[
  \D^Y=\{j:\, (i,j)\in \D \textrm{ for some } i\}.
\]
Given an elementary rectangle $R=[a,b]\times [c,d]$ of $Q_k$ and a positive integer $\ell\geq k$, we define the \emph{level-$\ell$ vertical partition} of $R$ (w.r.t. the Bedford-McMullen set $E$) to be
\begin{equation}\label{eq:V-ell-def}
  \V_\ell(R)=\left\{[a,b]\times [t,t+m^{-\ell}]:\, c\leq t \leq d-m^{-\ell}, t\in \D^Y_{\ell}\right\},
\end{equation}
where
$
  \D^Y_{\ell}=\{\sum_{j=1}^\ell \varepsilon_j m^{-j}:\, \varepsilon_1,\ldots,\varepsilon_\ell\in \D^Y\}.
$
Then $E\cap R=\bigcup_{A\in \V_\ell(R)} (E\cap A)$, and $E\cap A\not=\emptyset$ for every $A\in \V_\ell(R)$.

The next lemma proves the second part of Theorem~\ref{thm1}, and the upper bound estimation of $h_E(n^{-k})$ in the proof is essentially the same as that in \cite{MXX}.
\begin{lem}\label{lem:ERC-gap}
  Let $E$ be a Bedford-McMullen set satisfying the exponential rate condition. Then \eqref{eq:gap-McMullen} holds.
\end{lem}

\begin{proof}
In~\cite{RaRuY08}, Rao, Ruan and Yang showed  if two compact metric spaces
with infinite gap sequences are Lipschitz equivalent, their gap sequences are comparable. Let $\{\alpha_k\}_{k\ge1}$ and $\{\alpha'_k\}_{k\ge1}$ be the gap sequences of
$E$ with respect to the Euclidean distance and the maximum distance, respectively. Here the maximum distance $\rho(\cdot,\cdot)$ in $\mathbb{R}^2$ is defined as
\[
  \rho((x_1,y_1),(x_2,y_2))=\max\{|x_1-x_2|,|y_1-y_2|\},
\]
where  $(x_1,y_1),(x_2,y_2)\in \R^2.$ Then
\[
\alpha_k\asymp\alpha'_k.
\]
Therefore, in the rest of the proof, we assume that $\mathbb{R}^2$ is the space equipped with the maximum distance.


By Proposition~\ref{prop1},  \eqref{eq:gap-McMullen} is equivalent to $h_E(n^{-k})\asymp n^{\gamma k}$, where
\[\gamma=\begin{cases}
    \bdim E, & \text{if $M<m$}; \\
    \log_n N, & \text{if $M=m$}.
  \end{cases}\]

\textbf{Case~1: $M=m$.} It suffices to  show that
\[
   h_E(n^{-k})\asymp N^k.
\]

First, we prove that $h_E(n^{-k})\le
N^k$ for all $k\in \Z^+$. In order to prove this inequality,  it suffices to prove the following claim: for each elementary rectangle $R$ of $ Q_k$,
any two points in $E\cap R$ are $n^{-k}$-equivalent  in $E$.

For each given  integer $k$, let $\ell$ be the integer such that $m^{-\ell-1}<n^{-k}\le m^{-\ell}$. We split $R$
into $m^{\ell+2-k}$ rectangles with width $n^{-k}$ and height $m^{-\ell-2}$.
Since $M=m$, all the $m^{\ell+2-k}$ rectangles intersect $E$. We pick a point in $E\cap R$ from each small rectangle. Denote these points by $p_i=(x_i,y_i), i=1,\ldots,m^{\ell+2-k}$, where $y_1\leq y_2\leq \cdots \leq y_{m^{\ell+2-k}}$. Then the distances between any two adjacent points are bounded by
\[
  \rho(p_i,p_{i+1})\le\max\{n^{-k},2m^{-\ell-2}\}=n^{-k},
\]
for all $i$. Thus, $p_i$ and $p_j$ are $n^{-k}$-equivalent in $E\cap R$ for all $i$ and $j$.  This clearly implies that any two points in $E\cap R$ are $n^{-k}$-equivalent in $E$.

Next, we give the lower bound of $h_E(n^{-k})$.  Since $E$ satisfies the exponential rate condition, there exists a constant $c_0>0$, such that $\#\C(Q_k)\geq c_0 N^k$ for all $k\in \Z^+$.

For each integer $k\geq 2$, since $Q_{k-1}$ is the union of elementary rectangles with width $n^{-k+1}$ and height
$m^{-k+1}$, for any two distinct components $C,C'$ of $Q_{k-1}$ and any two points $p\in C,p'\in C'$, the distance between $p$ and $p'$ is bounded below by
\[
  \rho(p,p')\geq \min\{n^{-k+1},m^{-k+1}\}=n^{-k+1}>n^{-k}.
\]
Combining this with $E\subset Q_{k-1}$, if two points in $E$ belong to different components of $Q_{k-1}$, then they are not $n^{-k}$-equivalent in $E$.
Thus
\[
  h_E(n^{-k})\geq \#\C(Q_{k-1})\geq c_0N^{k-1}
\]
since every component of $Q_{k-1}$ intersects $E$.

\medskip
\textbf{Case~2: $M<m$.} Since $\bdim E=(\log N-\log M)/\log n+\log M/\log m$, it suffices to show that there exist two positive constants $c_1,c_2$, such that for all $k\in \Z^+$,
\[
  c_1N^{k}M^{\ell-k} \leq h_E(n^{-k}) \leq c_2 N^k M^{\ell-k},
\]
where $\ell$ is the unique integer satisfying $m^{-\ell-1}<n^{-k}\leq m^{-\ell}$.

First, we prove the upper bound estimate.
For each $k\in \Z^+$, there are $N^k$ elementary rectangles of $ Q_k$. For every elementary rectangle $R$ of $ Q_k$, by the definition of $M$, there exist $M^{\ell+1-k}$ small rectangles $R_1,\ldots,R_{M^{\ell+1-k}}$  in $\V_{\ell+1}(R)$,  such that every $R_i$ has width $n^{-k}$ and height
$m^{-\ell-1}$, and $R\cap E=\bigcup_{1\leq i\leq M^{\ell+1-k}} (R_i\cap E)$ .

The number of all such small rectangles is $N^{k} M^{\ell+1-k}$. Moreover,
any two points in the same small rectangle are $n^{-k}$-equivalent. Hence
$$
 h_E(n^{-k})\leq N^{k} M^{\ell+1-k}.
$$

Next, we prove the lower bound estimate. Since $E$ satisfies the exponential rate condition, there exists a constant $c_0>0$, such that $\#\C(Q_k)\geq c_0 N^k$ for all $k\in \Z^+$.

For each $k\geq 2$, there are $N^{k-1}$ elementary rectangles of $Q_{k-1}$. By \eqref{eq:V-ell-def} and the definition of $M$, for every elementary rectangle $R$ of $Q_{k-1}$, there exist $M^{\ell-k}$ rectangles $R_1,\ldots,R_{M^{\ell-k}}$ in $\V_{\ell-1}(R)$,
where $m^{-\ell-1}<n^{-k}\leq m^{-\ell}$, such that every $R_i$ has width $n^{-k+1}$ and height $m^{-\ell+1}$, and $R\cap E=\bigcup_{1\leq i\leq M^{\ell-k}} (R_i\cap E)$. Write $\widetilde{R}=\bigcup_{1\leq i\leq M^{\ell-k}} R_i$. It follows from  $M<m$ that  $\#\C(\widetilde{R})\geq M^{\ell-k-1}$. Meanwhile, if two points $p$ and $p'$ belong to distinct components of $\widetilde{R}$,  their distance is bounded below by
\[
   \rho(p,p')\geq \min\{ n^{-k+1}, m^{-\ell+1}\} >\min\{n^{-k}, m^{-\ell}\}= n^{-k}.
\]

Let $C$ be the component of $Q_{k-1}$ with $R\subset C$. From the construction of Bedford-McMullen set, if two points in $E\cap \widetilde{R}$ belong to distinct components of $\widetilde{R}$, they are not $n^{-k}$-equivalent in $E\cap C$. Thus, noticing that every component of $\widetilde{R}$ intersects $E\cap C$,
\[
  h_{E\cap C}(n^{-k}) \geq \#\C(\widetilde{R})\geq M^{\ell-k-1}.
\]

It is clear that if two points in $E$ belong to distinct components of $Q_{k-1}$, they are not $n^{-k}$-equivalent in $E$. Thus
\[
  h_E(n^{-k}) \geq \big(\#\C(Q_{k-1})\big) \cdot M^{\ell-k-1} \geq c_0 N^{k-1}M^{\ell-k-1}.
\]

This completes the proof of the lemma.
\end{proof}

\begin{proof}[Proof of Theorem~\ref{thm1}]
  The theorem follows from Lemmas~\ref{CSC-ERC} and \ref{lem:ERC-gap}.
\end{proof}

\section{Topological properties of Bedford-McMullen sets} \label{sec_tpg}
\setcounter{equation}{0}

First, we give some definitions. Let $R=[a,b]\times [c,d]$ be a rectangle in $\mathbb{R}^2$. We call $\{a\}\times [c,d]$, $\{b\}\times [c,d]$, $[a,b]\times \{c\}$, and $[a,b]\times \{d\}$ \emph{the left boundary, the right boundary, the bottom boundary, and the top boundary} of $R$, respectively.

Recall that $Q_0=[0,1]^2$ and for each integer $k\geq 1$,
\begin{eqnarray*}
  Q_{k}=\bigcup_{w\in \Sigma^k} \Psi_w(Q_0),  \qquad\qquad \wQ_k=\bigcup_{(i,j)\in \{-1,0,1\}^{2}} \Big(Q_k+(i,j)\Big).
\end{eqnarray*}
Let $C$ be a component of $Q_{k}$, we call $C$ \emph{vertical} if $C$ intersects both the top boundary and the bottom boundary of $Q_0$, i.e.
$$
   C\cap \big([0,1]\times \{0\} \big) \neq \emptyset, \quad \textit{and } \quad C\cap \big( [0,1]\times \{1\}\big) \neq \emptyset.
$$
Similarly, the component $ C$ is called \emph{horizontal} if $C$ intersects both the left boundary and the right boundary of $Q_0$. We call $Q_{k}$ \emph{vertical} (resp. \emph{horizontal}) if  all components  of $Q_{k}$   are \emph{vertical} (resp. \emph{horizontal}).

\medskip

The simple facts listed below are frequently used in our proof, where $E$ is the Bedford-McMullen set corresponding to the self-affine IFS $\{\Psi_w\}_{w\in \D}$.

\begin{enumerate}
\item[(i)] Assume that $A,B$ are two subsets of $\R^2$ with $A\subset B$. Then, given a component $C_A$ of $A$, there exists a component $C_B$ of $B$ such that $C_A\subset C_B$. Conversely, given a component $C_B$ of $B$ with $C_B\cap A\not=\emptyset$, there exists a component $C_A$ of $A$ such that $C_A\subset C_B$.


\item[(ii)] Let $k$ be a positive integer. Then $Q_k$ intersects the top boundary of $Q_0$ if and only if there exists $p\in \{0,1,\ldots,n-1\}$ such that $(p,m-1)\in \mathcal{D}$. Therefore, $Q_k$ intersects the top boundary of $Q_0$ if and only if $Q_1$ intersects the top boundary of $Q_0$. Analogous results hold for other boundaries.

\item[(iii)] Assume that $Q_1$ intersects neither the top boundary nor the right boundary of $Q_0$. Then $Q_1\cap \overline{(\wQ_{1}\backslash Q_1)}=\emptyset$. Consequently, $\{\Psi_w\}_{w\in \mathcal{D}}$ satisfies the \emph{strong separation  condition}, i.e., $\Psi_w(E)\cap \Psi_{u}(E)=\emptyset$ for all distinct $w,u\in \D$. Hence, the set $E$ is totally disconnected.
\item[(iv)] Let $k$ be a positive integer. Assume that $C_1$ and $C_2$ are  a vertical  component and  a horizontal  component of $Q_k$, respectively. Then $C_1=C_2$ since $C_1$ intersects $C_2$ in this case. As a result, if every component of $Q_k$ is either vertical or horizontal, then $Q_k$ is either vertical or horizontal.

\end{enumerate}

We call $E$ \emph{one-sided} if  one of the following conditions is satisfied:
\begin{enumerate}
  \item $\mathcal{D}\subset \{0,1,\ldots,n-1\}\times\{0\}$,
  \item $\mathcal{D}\subset \{0,1,\ldots,n-1\}\times\{m-1\}$,
  \item $\mathcal{D}\subset \{0\}\times\{0,1,\ldots,m-1\}$,
  \item $\mathcal{D}\subset \{n-1\}\times\{0,1,\ldots,m-1\}$.
\end{enumerate}

\medskip

The following topological lemma plays a crucial role in proving Theorem~\ref{Thm2}. We remark that in the proofs of the results of this section, we do not use the fact that $n\geq m$.

\begin{lem}\label{Tlem}
Let $E$ be a Bedford-McMullen set. Assume that $E$ is not one-sided,  and every  component of $E$  intersects the boundary of $Q_0$.
Then for every positive integer $k$, $Q_{k}$ is either vertical or horizontal.
\end{lem}

\begin{proof}
Given a positive integer $k$, arbitrarily choose a  component $C$ of $Q_k$. From Fact (i) and $E\subset Q_k$, there exists a component $C'$ of $E$ such that $C'\subset C$. Since every component in $E$ intersects the boundary of $Q_0$, clearly the component $C$ intersects the boundary of $Q_0$.

Without loss of generality,  we  assume that $C$ intersects the bottom boundary of $Q_0$ and does not intersect the top boundary of $Q_0$. We show that $C$ is horizontal in this case.
\medskip

First, we claim that $C$  intersects  either the  left boundary or the right boundary of $Q_0$. Otherwise, combining with the assumption that $C$ intersects the bottom boundary of $Q_0$, there exists a positive integer $1\leq i\leq n^{k}-2$ such that $T^{k}(Q_0)+(\frac{i}{n^{k}},0)\subset C$. Since  $E$ is not one-sided, there exist $q>0$ and $0\leq p\leq n-1$ such that $(p,q)\in \mathcal{D}$.

Let $C_1=(\frac{i}{n^{k}},0)+T^k(C)$. Then $C_1$ is a component of $Q_{2k}$ since $C$ only intersects the boundary of $Q_0$ at the bottom. Thus $(\frac{p}{n},\frac{q}{m})+T(C_1)\subset Q_{2k+1}$, and $T(C_1)$ only intersects the boundary of $(\frac{i}{n^{k+1}},0)+T^{k+1}( Q_{0})$ at the bottom. Combining this with $q>0$ and the fact that every component of $ Q_{2k+1}$  intersects the boundary of $Q_{0}$, we have $(p,q-1)\in\mathcal{D}$ and
 $T^{k}(Q_{0})+(\frac{i}{n^{k}},\frac{m^{k}-1}{m^{k}})\subset Q_{k}$. Thus
$$T^{k}(C)+\Big(\frac{i}{n^{k}},\frac{m^{k}-1}{m^{k}}\Big)
\subset T^{k}(Q_{k})+\Big(\frac{i}{n^{k}},\frac{m^{k}-1}{m^{k}}\Big)\subset Q_{2k}.$$
By the similar argument as above, $T^{k}(Q_{0})+(\frac{i}{n^{k}},\frac{m^{k}-2}{m^{k}})\subset Q_{k}$.
Repeating it finitely many times, we obtain that $T^{k}(Q_{0})+(\frac{i}{n^{k}},\frac{j}{m^{k}})\subset Q_{k}$ for all $0\leq j\leq m^{k}-1$. Thus $T^{k}(Q_{0})+(\frac{i}{n^{k}},\frac{j}{m^{k}})\subset C$ for all $0\leq j\leq m^{k}-1$, since $T^{k}(Q_0)+(\frac{i}{n^{k}},0)\subset C$.  Therefore, $ C$ is a vertical component in $ Q_k$, which contradicts that $C$ does not intersect the top boundary of $Q_0$.
\bigskip

Next, we show by contradiction that the component $C$ is horizontal.

Assume  that $C$ intersects the left boundary of $Q_0$ but does not intersect the right boundary of $Q_0$. Therefore, $C$ intersects both the bottom boundary and the left boundary of $Q_0$, whereas $C$  intersects neither the top boundary nor the right boundary of $Q_0$. We write
$$i_{0}=\max\big\{i:\,(\frac{i}{n^{k}},0)\in C \big\}, \quad j_{0}=\max\big\{j:\,(0,\frac{j}{m^{k}})\in C \big\}.$$
It is obvious that $i_{0}\leq n^{k}-1$, $j_{0}\leq m^{k}-1$, and there is a curve $\gamma_0$ in $C$ connecting the point $(\frac{i_{0}}{n^{k}},0)$ and the point $(0,\frac{j_{0}}{m^{k}})$. Please see Figure~\ref{figure: Cprime}.

Now, we will prove the following fact:
\begin{equation}\label{Tlem9}
 Q_{k}\cap\Big(\Big[\frac{i_{0}+1}{n^{k}},1\Big]\times\{0\}\Big)\neq\emptyset \quad \text{or}\quad  Q_{k}\cap\Big(\{0\}\times\Big[\frac {j_{0}+1}{m^{k}},1\Big]\Big)\neq\emptyset.
\end{equation}
Since every component of $E$  intersects the boundary of $Q_0$ and $E$ is not one-sided, by Facts (ii) and (iii), the set $ Q_{k}$ intersects either the right boundary or the top boundary of $Q_{0}$. Without loss of generality, we assume that $Q_k$ intersects the right boundary of $Q_0$. Then there exists
$0\leq \ell \leq m^{k}-1$  such  that  $(\frac{n^{k}-1}{n^{k}},\frac{\ell}{m^{k}})+T^{k}( Q_{0})\subset Q_{k}$.

In the case that $\ell=0$, it is easy to see $(1,0)\in Q_{k}\cap\big(\big[\frac{i_{0}+1}{n^{k}},1\big]\times\{0\}\big)$ so that \eqref{Tlem9} holds. Now we assume that $\ell>0$. It follows from $(\frac{n^{k}-1}{n^{k}},\frac{\ell}{m^{k}})+T^{k}( Q_{0})\subset Q_{k}$ that $(\frac{n^{k}-1}{n^{k}},\frac{\ell}{m^{k}})+T^{k}( Q_{k})\subset Q_{2k}$.  Recall that every component of $E$ intersects the boundary of $Q_0$. Since $(\frac{n^{k}-1}{n^{k}},\frac{\ell}{m^{k}})+T^{k}( C)$ only intersects the bottom boundary and  the left  boundary of $(\frac{n^{k}-1}{n^{k}},\frac{\ell}{m^{k}})+T^{k}(Q_{0})$, at least one of the following three sets is contained in $Q_{2k}$:
\[
 \Big(\frac{n^{k}-2}{n^{k}},\frac{\ell}{m^{k}}\Big)+T^{k}( Q_{k}), \Big(\frac{n^{k}-1}{n^{k}},\frac{\ell-1}{m^{k}}\Big)+T^{k}(Q_{k}), \Big(\frac{n^{k}-2}{n^{k}},\frac{\ell-1}{m^{k}}\Big)+T^{k}(Q_{k}).
\]
Repeating this argument, one of the followings must hold:
\begin{enumerate}
\item[(a).] there exist $0\leq n'\leq n^k$ and a curve $\gamma_1$ in $Q_k$ joining  $(\frac{n^{k}-1}{n^{k}},\frac{\ell}{m^{k}})$  and $(\frac{n'}{n^k},0)$;
\item[(b).] there exist $0\leq m^\prime\leq m^k$ and a curve $\gamma_2$ in $Q_k$ joining  $(\frac{n^{k}-1}{n^{k}},\frac{\ell}{m^{k}})$  and $(0,\frac{m^\prime}{m^k})$.
\end{enumerate}
In case (a), $n'\geq i_0+1$, since otherwise the curve $\gamma_1$ intersects $\gamma_0$ so that $(\frac{n^{k}-1}{n^{k}},\frac{\ell}{m^{k}})+T^{k}(Q_{0})\subset C$, which contradicts the assumption that $C$ does not intersect the right boundary of $Q_0$. Similarly, $m^\prime \geq j_0+1$ in case (b). Hence fact \eqref{Tlem9} holds.
\medskip

By fact \eqref{Tlem9}, we assume that $ Q_{k}\cap\big([\frac{i_{0}+1}{n^{k}},1]\times\{0\}\big)\neq\emptyset$.
We write
\begin{equation*}
  i^{*}=\min\big\{i\in\{i_0+1,\ldots,n^k-1\}:\,( \frac{i}{n^{k}},0)+T^{k}( Q_{0})\subset Q_{k}\big\}.
\end{equation*}
Then $C'=(i^*/n^k,0)+T^k(C)$ is a component of $Q_{2k}$ and
\begin{equation}\label{eq:4-5}
  \max\Big\{i:\, \big(\frac{i^*}{n^k}+\frac{i}{n^{2k}},0\big)+T^{2k}(Q_0)\subset C'\Big\} = i_0-1\leq i^*-2.
\end{equation}
Please see Figure~\ref{figure: Cprime}.

\begin{figure}[tbhp]
\begin{center}
  \includegraphics[height=5cm]{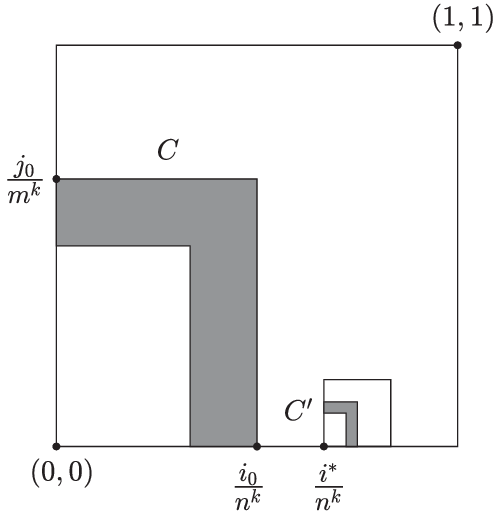}
\end{center}
\caption{Components $C$ and $C^\prime$}
\label{figure: Cprime}
\end{figure}

Recall that there exist $q>0$ and $0\leq p\leq n-1$ such that $(p,q)\in \mathcal{D}$.

If $(\frac{i}{n^k}, \frac{m^k-1}{m^k})+T^k(Q_0)\not\subset Q_k$ for all $0\leq i\leq i^*-1$, then
\[
  \Big(\frac{i^*}{n^k}+\frac{i}{n^{2k}}, \frac{m^{2k}-1}{m^{2k}}\Big)+T^{2k}(Q_0)\not\subset Q_{2k}, \quad -1\leq i\leq i^*-1.
\]
Combining this with \eqref{eq:4-5}, the set $(\frac{p}{n},\frac{q}{m})+T(C')$ does not  intersect  $(\frac{p}{n},\frac{q-1}{m})+T(Q_{2k})$ so that it
is a component of $ Q_{2k+1}$. However, $(\frac{p}{n},\frac{q}{m})+T( C')$ does not intersect the  boundary  of $Q_{0}$, which  contradicts the  fact that every  component of $E$  intersects the boundary of $Q_0$.

Now we assume that there exists $0\leq i^\prime\leq i^*-1$, such that $(\frac{i^\prime}{n^k},\frac{m^k-1}{m^k})+T^k(Q_0)\subset Q_k$. Using the same argument in proving \eqref{Tlem9} and noticing that $(\frac{i}{n^k},0)+T^k(Q_0)\not\subset Q_{k}$ for all $i_0\leq i\leq  i^*-1$, there exists $j_0+1\leq j'\leq m^k-1$, such that
\[
  \Big(0,\frac{j'}{m^k}\Big) + T^k(Q_0) \subset Q_k.
\]
Let
\begin{equation*}
j^{*}=\min\big\{j\in\{j_0+1,\ldots,m^k-1\}:\,(0,\frac{j}{m^{k}})+T^{k}( Q_{0})\subset Q_{k}\big\}.
\end{equation*}
Then $C''=(0,\frac{j^*}{m^k})+T^k(C)$ is a component of $Q_{2k}$. Thus $(\frac{i^*}{n^k},0)+T^k(C'')$ is a component of $Q_{3k}$, and it does not intersect the  boundary  of $Q_{0}$, which gives a contradiction again.

From the above argument,  $C$ is a horizontal component. Thus every component of $Q_k$ is either vertical or horizontal. By Fact (iv), the lemma holds.
\end{proof}

\begin{lem}\label{Thm}
Let $E$ be a Bedford-McMullen set which is not  one-sided. If  $E$ is normal and every component of $E$  intersects the boundary of $Q_0$, then $E$ contains only finitely many components.
\end{lem}

\begin{proof}
Since $E$ is normal and every component of $E$  intersects the boundary of $Q_0$, by Lemma~\ref{Tlem},  for every integer $k>0$, $Q_{k}$ is either vertical or horizontal. Without loss of generality, we assume that there exist infinitely many integers $k>0$ such that the sets $Q_k$ are vertical. By Fact (i), $Q_k$ is vertical for all $k>0$.

First, we define an equivalence relation on $\mathcal{D}$. Given $(i,j),(i^{\prime},j^{\prime})\in\mathcal{D}$, we say that $(i,j)$ and $(i^{\prime},j^{\prime})$ are \emph{vertically equivalent with respect to $\mathcal{D}$}, if $i=i^{\prime}$ and $(i,\ell)\in\mathcal{D}$ for all $\ell$ lying between $j$ and $j'$.  We denote the vertically equivalent classes by $\Lambda_{1},\Lambda_{2},\ldots,\Lambda_{\beta}$. Without loss of generality, we assume that there exists  $\beta_0\in\{1,2,\ldots,\beta\}$ such that
\begin{align*}
 & \Lambda_{i}\cap\big( \{0,1,\ldots,n-1\}\times\{0\}\big)\neq\emptyset, \qquad  i=1,2,\ldots, \beta_0, \qquad  \mbox{and}  \\
 & \Lambda_{i}\cap\big( \{0,1,\ldots,n-1\}\times\{0\}\big)=\emptyset, \qquad i= \beta_0+1, \ldots, \beta.
\end{align*}
 We write $\lambda_{i}=\#\Lambda_{i}$ for each $1\leq i\leq\beta$, and define $\lambda=\min\{\lambda_1,\ldots,\lambda_\beta\}$. Then $\lambda<m$ since $E$ is normal.

In the proof of this lemma, for a subset $A$ of $\mathbb{R}^2$, we denote by $\C^B(A)$ the set of all components of $A$ which intersect the bottom boundary of $Q_0$.

Given  positive integers $k,\ell$, we define
\begin{equation*}
 Q_{k,\ell}= Q_{k}\cup( Q_{k}+(0,1))\cup\cdots\cup( Q_{k}+(0,\ell-1)).
\end{equation*}
Notice that $\lambda_{i}\geq \lambda$ for all $1\leq i\leq \beta$. Thus, given a component $C$ in $Q_{k, \lambda_i}$ with $C$ intersecting the bottom boundary of $Q_0$, there exists a component $C'$ in $Q_{k,\lambda}$ such that $C'\subset C$ and $C'$ intersects the bottom boundary of $Q_0$. Hence $\#\C^B( Q_{k,\lambda_{i}})\leq \#\C^B( Q_{k,\lambda})$.

Notice that
\[
  Q_{k+1}=\bigcup_{w\in \D} \Psi_w(Q_k) = \bigcup_{i=1}^\beta\bigcup_{w\in \Lambda_i} \Psi_w(Q_k),
\]
while $\bigcup_{i=\beta_0+1}^\beta\bigcup_{w\in \Lambda_i} \Psi_w(Q_k)$ does not  intersect the bottom boundary of $Q_0$. Since every component in $Q_{k+1}$  intersects the bottom boundary of $Q_0$, each of them must contain a component in $\bigcup_{i=1}^{\beta_0}\bigcup_{w\in \Lambda_i}\Psi_w(Q_k)$ intersecting the bottom boundary of $Q_0$. Thus
\[
  \#\C(Q_{k+1}) \leq \# \C^B\Big(\bigcup_{i=1}^{\beta_0} \bigcup_{w\in \Lambda_i} \Psi_w(Q_k)\Big) \leq \sum_{i=1}^{\beta_0} \# \C^B \Big(\bigcup_{w\in \Lambda_i} \Psi_w(Q_k)\Big).
\]
Combining this with the fact that $\#\C^B\big(\bigcup_{w\in \Lambda_i} \Psi_w(Q_k) \big) = \# \C^B(Q_{k,\lambda_i}) $ for all $1\leq i\leq \beta_0$, we have

\begin{eqnarray}\label{eq:4-8}
  \#\C( Q_{k+1}) \leq \sum_{i=1}^{\beta_0} \#\C^B( Q_{k,\lambda_{i}})
                  \leq \beta_0 \cdot \big(\#\C^B( Q_{k,\lambda})\big).
\end{eqnarray}

 Next, we show  that  $\#\C^B( Q_{k,\lambda})\leq 2\lambda$ for every positive integer $k$.

 Let $\partial_{LR} (A)$ be the union of the left boundary and the right boundary of a rectangle $A$. Since $ Q_{k}$ is vertical, we have
 \[
   \#\{C\in \C(Q_k):\, C \cap \partial_{LR}(Q_0) \not=\emptyset \} \leq 2.
 \]
 Hence, if we write $Q_{0,\lambda}=[0,1]\times [0,\lambda]$, then
 \[
   \#\{C\in \C(Q_{k,\lambda}):\, C \cap \partial_{LR}(Q_{0,\lambda}) \not=\emptyset \} \leq 2\lambda.
\]
Thus in order to prove  $\#\C^B( Q_{k,\lambda})\leq 2\lambda$, it suffices to show that every component of $Q_{k,\lambda}$ must intersect  either the left boundary or the right boundary of  $Q_{0,\lambda}$. We prove this by contradiction.

 Assume that there exists a component $ C$ of $Q_{k,\lambda}$ such that $C$ intersects neither the left boundary nor the right boundary of $Q_{0,\lambda}$. According to the definition of  $\lambda$, there exists $(p,q)\in \mathcal{D}$ such that $\{(p,q),(p,q+1),\ldots,  (p,q+\lambda-1) \}$ is a vertically equivalent class. Notice that $(\frac{p}{n},\frac{q}{m})+T( Q_{k,\lambda})\subset  Q_{k+1}$. Write
\[
  C'=(\frac{p}{n},\frac{q}{m})+T( C).
\]
Then $C'\subset (\frac{p}{n},\frac{q}{m})+T( Q_{k,\lambda})$ and
$C'$ intersects neither the left boundary nor the right boundary of  $(\frac{p}{n},\frac{q}{m})+T( Q_{0,\lambda})$. Thus $C'$ is  a component of $ Q_{k+1}$. Since $\lambda<m$,  the component
$C'$ can intersect at most one of the top boundary and the bottom  boundary of $Q_{0}$, which implies that $Q_{k+1}$ is not vertical. This contradicts the fact that  $Q_k$ is vertical for all $k>0$.

Combining \eqref{eq:4-8} with $\#\C^B( Q_{k,\lambda})\leq 2\lambda$ for all $k\geq 1$, we have  $\#\C( Q_{k+1})\leq 2\beta_0 \lambda.$
Since both $\beta_0$ and  $\lambda$ are independent of $k$,  $\#\C( Q_{k})\leq 2\beta_0 \lambda$ for all $k\geq 2$. Hence there are at most $2\beta_0 \lambda$ components of $E$.
\end{proof}

\begin{lem}\label{Tlem4}
Let $E$ be a Bedford-McMullen set containing  infinitely many components.  If $E$ is  one-sided,  then $E$ satisfies the component separation condition.
\end{lem}
\begin{proof}
Without loss of generality, we  assume that $\mathcal{D}\subset \{0,1,\ldots,n-1\}\times\{0\}$.  Since $E$ contains  infinitely many components, $\D\not=\{0,1,\ldots,n-1\}\times\{0\}$.

If there exists a  component $ C$ of $ Q_{1}$ such that $C$ intersects neither the left boundary nor the right boundary of $Q_{0}$,  it is clear that $E$ satisfies the component separation condition.

Otherwise all components $C$ of $ Q_{1}$ intersect either the left boundary or the right boundary of  $Q_{0}$. Then $\#\C(Q_1)=1$ or $2$ since $E$ is one-sided. Obiviously, $E$ satisfies the component separation condition if $\#\C(Q_1)=1$.
Now we assume that $\#\C(Q_1)=2$. Then there exist $0\leq p<q\leq n-1$, such that $\mathcal{D}=\{(i,0)|\, 0\leq i\leq p, \textrm{ or } q\leq i\leq n-1\}$  with $p<q-1$. Write
\[
  C=\Big\{0,\frac{1}{n},\ldots,\frac{p}{n}\Big\} \times \{0\} + T(Q_0).
\]
Then the set $(q/n,0)+T(C)$ is both a component of $ Q_{2}$ and a component of $\widetilde{Q}_2$. Thus $E$ satisfies the component separation condition.
\end{proof}

\begin{lem}\label{Thmmm9}
Let $E$ be a simple Bedford-McMullen set  with infinitely many connected components. Then
$$ g_k(E)\asymp k^{-1/(\bdim E-1)}.$$
\end{lem}

\begin{proof}  First we claim that if $K=K(n, \A)$ is the unique compact subset of $[0,1]$ satisfying $K=\frac{1}{n}(K+\A)$, where $\A$ is a subset of $\{0,1,\ldots,n-1\}$ with $2\leq \#\A \leq n-1$, then $g_k(K)\asymp k^{-1/ \bdim K}$.

We write $n_{0}=\#\A$ and assume that $\A=\{i_1,\ldots,i_{n_0}\}$.
 Let $\A'=\{j_1,\ldots,j_{n_0}\}$ be a proper subset of $\{0,1,\ldots,n-1\}$. From~\cite{RRX06}, we know that $K(n,\A')$ is Lipschitz equivalent to $K(n,\A)$. Thus, from~\cite{RaRuY08}, their gap sequences are comparable. Hence we assume without loss of generality that
  $\A=\{0,1,\ldots,n_0-1\}.$

Let  $\delta_1>\delta_2>\cdots $ be the discontinuity points of $h_{K}(\delta)$. Then
\[
  \delta_{j}=\frac{n-n_{0}}{(n-1)n^{j}}, \qquad h_{K}(\delta_{j})=n_{0}^{j-1}, \qquad j=1,2,\ldots.
\]
For every $k\in \Z^+$, we choose $j\in \Z^+$ such that $n_{0}^{j-1}\leq k< n_{0}^{j}.$ From \eqref{eq:gap-prop}, $g_k(K)=\delta_j$.
It is well-known that $\bdim K=\frac{\log n_0}{\log n}$,  see \cite{Falco03} for example. Thus $n^j= (n_0^j)^{1/\bdim K} \asymp k^{1/\bdim K}$ giving
\[
  g_k(K)= \frac{n-n_{0}}{(n-1)n^{j}} \asymp n^{-j} \asymp k^{-1/\bdim K}.
\]
This completes the proof of the claim.

Now we prove the lemma. Since $E$ is simple and contains infinitely many connected components, we may assume without loss of generality that $E=K(n,\A)\times [0,1]$ for some $\A\subset \{0,1,\ldots,n-1\}$ with $2\leq \#\A\leq n-1$. Notice that $E$ and $K(n,\A)$ have the same gap sequence, and $\bdim E=1+\bdim K(n,\A)$. The lemma follows from the claim.
\end{proof}

\medskip

\begin{proof}[Proof of Theorem~\ref{Thm2}]
  From Theorem~\ref{thm1} and Lemma~\ref{Thmmm9}, it suffices to show that if $E$  is normal, then $E$ satisfies the component separation condition.

If $E$ is  one-sided, by Lemma~\ref{Tlem4},  the component separation condition holds.  Now we assume that $E$ is not one-sided. By  Lemma~\ref{Thm},  there  exists a connected component $C$ of $E$  such that $C$ does not intersect the boundary of $Q_0$.  By Fact (i), for each $k\in \Z^+$, there exists a component $C_k$ in $Q_k$ containing $C$. Since $C_k$ is decreasing with respect to $k$, by a well-known topological result, $\bigcap_{k=1}^\infty C_k$ is connected, for examples see \cite[Exercise 11 in Section 26]{Munkres}. Thus $\bigcap_{k=1}^\infty C_k$ is a connected subset of $E$. Combining this with $C\subset \bigcap_{k=1}^\infty C_k$ and with the fact that $C$ is the component of $E$, we have $C= \bigcap_{k=1}^\infty C_k$.  Thus there exists $k_0\in \Z^+$ such that for all $k\geq k_0$, $C_k$ does not intersect the boundary of $Q_0$. This implies that the component separation condition holds.
\end{proof}

\section{Application to Lipschitz equivalence}
Recently there has been some progress on Lipschitz equivalence of totally disconnected Bedford-McMullen sets \cite{LiLiMi13, RaYaZh20, YangZhang20}. We apply our result to a more  general setting. Fix integers $m$ and $n$ such that $n> m\geq 2$. Let $\mathcal{D}_1, \mathcal{D}_2$ be subsets of $\{0,\dots,n-1\}\times
\{0,\dots,m-1\}$ with $1<\#\D_1=\#\D_2<mn$. Let $E_i=E(n,m,\D_i)$, $i=1,2$, be the Bedford-McMullen sets corresponding to $n,m$ and $\D_i$.  Let $N$ be the common value of $\#\D_1$ and $\#\D_2$. Write
\begin{equation}\label{eq:M1M2def}
  M_k=\#\{j\colon (i,j)\in \mathcal{D}_k \textrm{ for some } i\}, \quad k=1,2.
\end{equation}
From \eqref{eq:DboxFormula}, it is easy to see that
\begin{equation}\label{eq:5-5}
  \bdim E_1=\bdim E_2  \textrm{ if and only if } M_1=M_2.
\end{equation}

Assume that exactly one of $E_1$ and $E_2$ is simple. We claim that their gap sequences are not comparable. To prove this, it suffices to consider the case that both $E_1$ and $E_2$ contain infinitely many components.
Assume that $E_1$ is simple and $E_2$ is normal. Notice that
\begin{align}
  &\bdim E_1-1=\frac{\log N-\log M_1}{\log n} + \frac{\log M_1}{\log m} -1 \leq \frac{\log N-\log M_1}{\log n} < \frac{\log N}{\log n},  \notag \\
  &\bdim E_2=\frac{\log N-\log M_2}{\log n} + \frac{\log M_2}{\log m} > \frac{\log N}{\log n}. \label{eq:5-6}
\end{align}
By Theorem~\ref{Thm2}, $\{g_k(E_1)\}$ and $\{g_k(E_2)\}$ are not comparable. Thus $E_1$ and $E_2$ are not Lipschitz equivalent.

Assume that both $E_1$ and $E_2$ are simple. We claim that their gap sequences are comparable if and only if they are Lipschitz equivalent. To prove this, it suffices to consider the case that both $E_1$ and $E_2$ contain infinitely many components. From Theorem~\ref{Thm2},  in this case,  $\{g_k(E_1)\}$ and $\{g_k(E_2)\}$ are comparable if and only if $\bdim E_1=\bdim E_2$. By \eqref{eq:5-5}, both of them are equivalent to $M_1=M_2$. Thus, $\{g_k(E_1)\}$ and $\{g_k(E_2)\}$ are comparable if and only if there exist self-similar sets $K_1$ and $K_2$ such that one of the followings holds:
\begin{itemize}
  \item[(1).] $E_i=[0,1]\times K_i$, $i=1,2$;
  \item[(2).] $E_i=K_i\times [0,1]$, $i=1,2$.
\end{itemize}
By using the result in \cite{RRX06}, if one of the above conditions holds, then $K_1$ and $K_2$ are Lipschitz equivalent so that $E_1$ and $E_2$ are Lipschitz equivalent. Thus the claim holds.

Assume that both $E_1$ and $E_2$ are normal and contain infinitely many components. We claim that their gap sequences are comparable if and only if they have the same box dimension. From \eqref{eq:5-6}, $\bdim E_2>\frac{\log N}{\log n}$. Thus, if $M_1=m$ and $M_2<m$, then $\{g_k(E_1)\}$ and $\{g_k(E_2)\}$ are not comparable. Combining this result with  Theorem~\ref{Thm2} and \eqref{eq:5-5}, we know the claim holds. We remark that a similar argument was given in \cite{RaYaZh20} for the totally disconnected case.

It is well-known that most of dimensions, including box, packing, Hausdorff, Assouad, lower and intermediate dimensions, are stable under bi-Lipschitz maps. Thus, it is natural to ask the following question:
\begin{ques}\label{ques:1}
Can we construct two Bedford-McMullen sets $E_1$ and $E_2$ such that $\dim E_1= \dim E_2$ for all these dimensions, while their gap sequences are not comparable?
\end{ques}
There are explicit formulas for the box, packing, Hausdorff, Assouad and lower dimensions of Bedford-McMullen sets, while the box dimension and the packing dimension always coincide. For example, see Theorem~2.1 in \cite{Fraser20} for these formulas.  However, in general, it is not easy to obtain the intermediate dimension of Bedford-McMullen sets (see \cite{BanKol21}). Thus, in this paper, we only consider the box, Hausdorff, Assouad and lower dimensions.

Given a Bedford-McMullen set $E$, we say $E$ has \emph{uniform fibres} if all non-empty rows contain the same number of rectangles.  Notice that there is a dichotomy (see \cite{Fraser20}),  in the uniform fibres case,
\[
  \dim_{\mathrm L} E=\dim_{\mathrm H} E=\dim_{\mathrm B} E=\dim_{\mathrm A} E,
\]
and, in the non-uniform fibres case,
\[
  \dim_{\mathrm L} E<\dim_{\mathrm H} E<\dim_{\mathrm B} E<\dim_{\mathrm A} E.
\]

By using the above argument, it is quite easy to answer Question~\ref{ques:1} positively if exactly one of $E_1$ and $E_2$ is simple and both of them have uniform fibres.
\begin{exmp}\label{ex:LipUniFibre}
Let $E_1=E(3,2,\D_1)$ and $E_2=E(3,2,\D_2)$ with
\[
  \D_1=\{(0,0),(2,0),(0,1),(2,1)\}, \quad \D_2=\{(0,0),(1,0),(0,1),(2,1)\}.
\]
Then $\dim E_1=\dim E_2=1+\log_3 2$ for the box, Hausdorff, Assouad and lower dimensions, while their gap sequences are not comparable. See Figure~\ref{Fig:LipUniE12} for their initial structures.
\end{exmp}

  \begin{figure}[tbhp]
  \subfigure[$E_1$]{
    \begin{minipage}[t]{0.45\linewidth}
      \begin{center}
           {\includegraphics[height=4.0cm]{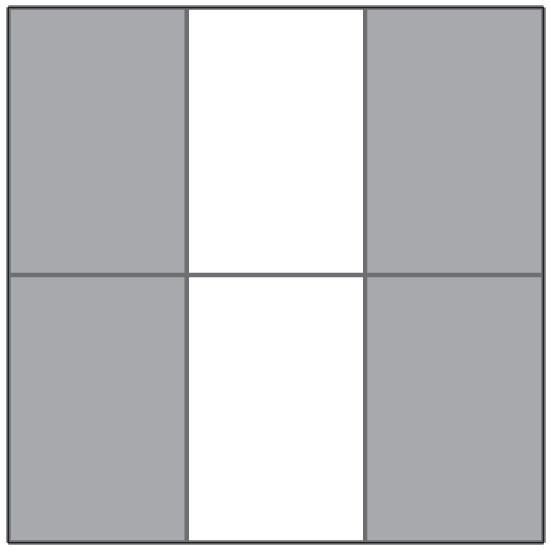}}
      \end{center}
  \end{minipage}
  }
  \subfigure[$E_2$]{
    \begin{minipage}[t]{0.45\linewidth}
       \begin{center}
           {\includegraphics[height=4.0cm]{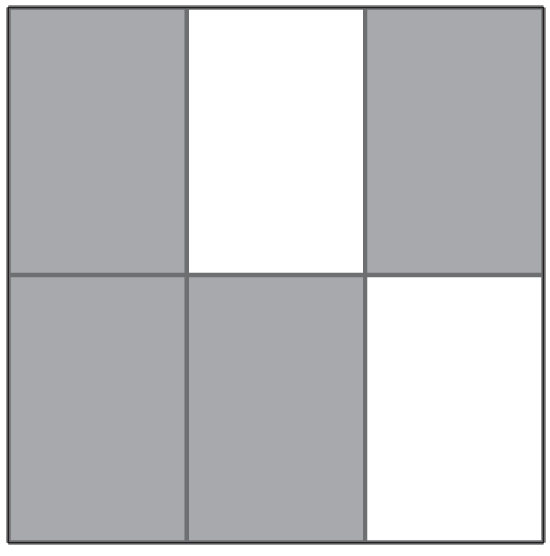}}
       \end{center}
    \end{minipage}
  }
  \caption{Initial structures of $E_1$ and $E_2$ in Example~\ref{ex:LipUniFibre}.}
  \label{Fig:LipUniE12}
 \end{figure}

Things become more interesting and also more complicated if we require that neither of them has uniform fibres. Fortunately, in this case, we can still answer Question~\ref{ques:1} positively.



Let $n=m^2$. Let $\mathcal{D}_1, \mathcal{D}_2$ be subsets of $\{0,\dots,n-1\}\times
\{0,\dots,m-1\}$ with $1<\#\D_1<\#\D_2<mn$, and satisfying the  following four conditions:
\begin{enumerate}
  \item $M_1\cdot \# \D_1 = M_2 \cdot \# \D_2$ and $M_1=m$, where $M_1$ and $M_2$ are defined as in \eqref{eq:M1M2def};
  \item $\sum_{j=0}^{m-1} \sqrt{M_{j,1}} = \sum_{j=0}^{m-1} \sqrt{M_{j,2}}$ , where
  \[
    M_{j,k} = \#\{i:\, (i,j)\in \D_k\}, \quad 0\leq j\leq m-1, \; k=1,2;
  \]
  \item $n\cdot \max_{0\leq j\leq m-1} M_{j,1}=(M_2)^2 \cdot \max_{0\leq j\leq m-1} M_{j,2}$;
  \item $n\cdot \min_{0\leq j\leq m-1} M_{j,1}=(M_2)^2 \cdot \min_{0\leq j\leq m-1} M_{j,2}$.
\end{enumerate}
Let $E_i=E(n,m,\D_i)$, $i=1,2$. Write $N_i=\# \D_i$, $i=1,2$. Then
\begin{align*}
  \bdim E_1=\frac{\log N_1-\log M_1}{\log n} + \frac{\log M_1}{\log m}  = \frac{\log (M_1N_1) }{2\log m}  = \frac{\log (M_2N_2) }{2\log m} =\bdim E_2.
\end{align*}
Furthermore, from Theorem~2.1 in \cite{Fraser20},
\begin{align*}
  &\dim_{\mathrm H} E_1 = \log_m \Big( \sum_{j=0}^{m-1} \sqrt{M_{j,1}}\Big)=\log_m \Big( \sum_{j=0}^{m-1} \sqrt{M_{j,2}}\Big)=\dim_{\mathrm H} E_2,\\
   & \dim_{\mathrm A} E_1=1+\frac{\log \Big(\max\limits_{0\leq j\leq m-1} M_{j,1} \Big) }{\log n} = \frac{\log M_2 }{\log m} + \frac{\log \Big(\max\limits_{0\leq j\leq m-1} M_{j,2} \Big) }{\log n}=\dim_{\mathrm A} E_2, \\
  & \dim_{\mathrm L} E_1=1+\frac{\log \Big(\min\limits_{0\leq j\leq m-1} M_{j,1} \Big) }{\log n} = \frac{\log M_2 }{\log m} + \frac{\log \Big(\min\limits_{0\leq j\leq m-1} M_{j,2} \Big) }{\log n}=\dim_{\mathrm L} E_2.
\end{align*}
Assume that both $E_1$ and $E_2$ are normal and contain infinitely many components. Then from Theorem~\ref{Thm2}, $g_k(E_1)\asymp k^{-\log n/\log N_1}$ and $g_k(E_2)\asymp k^{-1/\bdim E_2}$. Since
\[
  \bdim E_2 = \bdim E_1 >\frac{\log N_1}{\log n},
\]
$\{g_k(E_1)\}$ and $\{g_k(E_2)\}$ are not comparable.

\begin{exmp}\label{ex:Lipschitz}
  Let $m=4$, $n=16$, and let
  \begin{align*}
    M_{0,1}=M_{1,1}>M_{2,1}=M_{3,1}>0, \quad M_{0,2}>M_{3,2}>M_{1,2}=M_{2,2}=0
  \end{align*}
  with $M_{0,1}=4M_{0,2}$ and $M_{3,1}=4M_{3,2}$. Then it is easy to check that the four conditions to ensure the same dimensions are all satisfied.
  Thus, if we let $E_1=E(16,4,\D_1)$ and $E_2=E(16,4,\D_2)$ with
  \begin{align*}
    &\D_1=\{(i,j):\, i=0,4,8,12;j=0,1\} \cup \{(i,j):\, i=0,15;j=2,3\}, \\
    &\D_2=\{(i,0):\, 0\leq i\leq 15\} \cup \{(i,3):\, 0\leq i\leq 3, \text{or } 12\leq i\leq 15 \}.
  \end{align*}
  Then $\dim E_1=\dim E_2$ for the box, Hausdorff, Assouad and lower dimensions, while their gap sequences are not comparable. See Figure~\ref{Fig:LipE12} for their initial structures.

  Similarly, if we let $E_3=E(16,4,\D_3)$ and $E_4=E(16,4,\D_4)$  with
  \begin{align*}
    &\D_3=\{(0,0),(15,0),(0,1),(15,1),(8,2),(8,3)\}, \\
    &\D_4=\{(i,0):\, 0\leq i\leq 3, \text{or } 12\leq i\leq 15\} \cup \{(i,3):\, i=0,1,14,15\}.
  \end{align*}
  Then $\dim E_3=\dim E_4$ for the box, Hausdorff, Assouad and lower dimensions, while their gap sequences are not comparable. See Figure~\ref{Fig:LipE34} for their initial structures.

\end{exmp}



  \begin{figure}[tbhp]
  \subfigure[$E_1$]{
    \begin{minipage}[t]{0.45\linewidth}
      \begin{center}
           {\includegraphics[height=4.0cm]{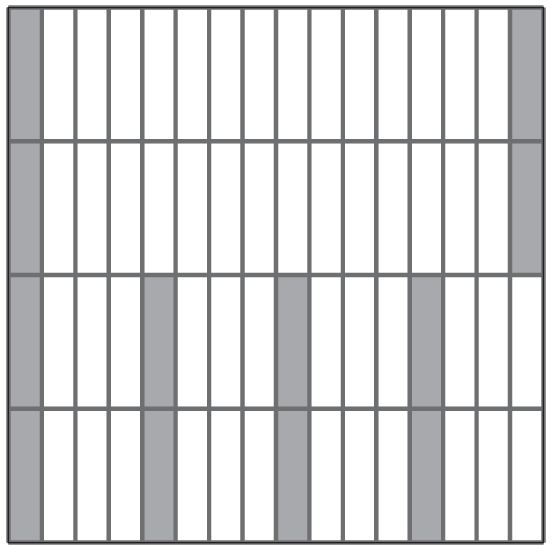}}
      \end{center}
  \end{minipage}
  }
  \subfigure[$E_2$]{
    \begin{minipage}[t]{0.45\linewidth}
       \begin{center}
           {\includegraphics[height=4.0cm]{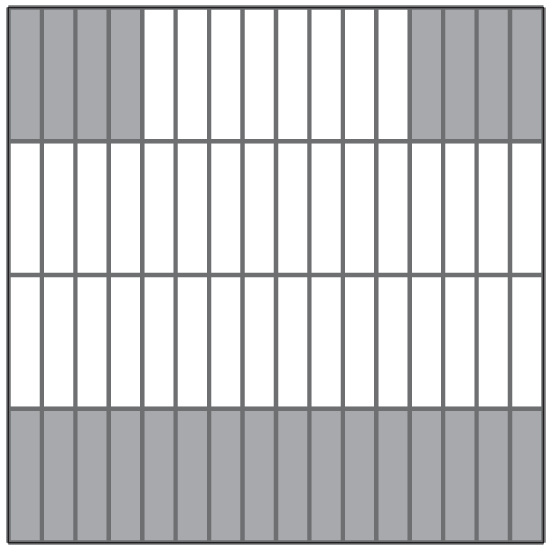}}
       \end{center}
    \end{minipage}
  }
  \caption{Initial structures of $E_1$ and $E_2$ in Example~\ref{ex:Lipschitz}.}
  \label{Fig:LipE12}
 \end{figure}

  \begin{figure}[tbhp]
  \subfigure[$E_3$]{
    \begin{minipage}[t]{0.45\linewidth}
      \begin{center}
           {\includegraphics[height=4.0cm]{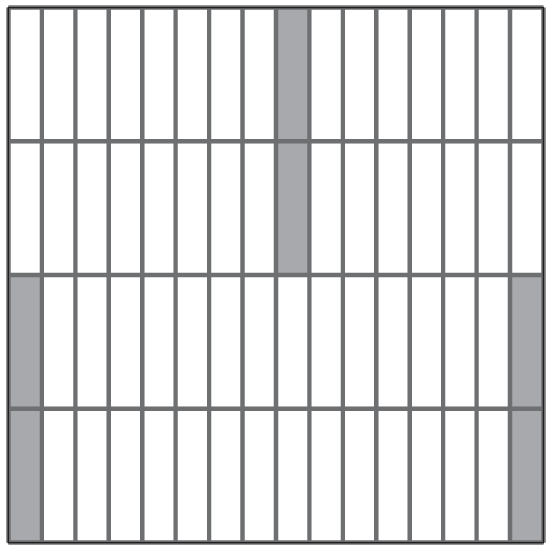}}
      \end{center}
  \end{minipage}
  }
  \subfigure[$E_4$]{
    \begin{minipage}[t]{0.45\linewidth}
       \begin{center}
           {\includegraphics[height=4.0cm]{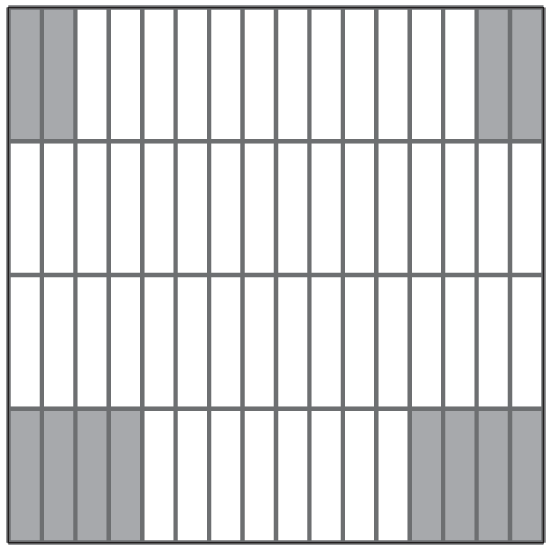}}
       \end{center}
    \end{minipage}
  }
  \caption{Initial structures of $E_3$ and $E_4$  in Example~\ref{ex:Lipschitz}.}
  \label{Fig:LipE34}
 \end{figure}

\section{Further remarks}

\subsection{The component separation condition and related conditions}

In this paper, we introduce two conditions for Bedford-McMullen sets: the CSC and the ERC. In Lemma~\ref{CSC-ERC}, we show that the CSC implies the ERC. Naturally we may ask whether the ERC also implies the CSC. It is easy to check that if $E$ is simple or $E$ contains finitely many components, then $E$ does not satisfy the ERC. Thus, it suffices to consider the case that $E$ is a normal Bedford-McMullen set with infinitely many components. In Section~4, we show that in this case, $E$ always satisfies the CSC.  Thus, the ERC also implies the CSC so that they are equivalent, and both of them are equivalent to the condition that $E$ is a normal Bedford-McMullen set with infinitely many components.

We  compare the CSC with other familiar separation conditions. Let $E$ be a Bedford-McMullen set. Then $E$ always satisfies the open set condition because of the grid structure. If $E$ satisfies the strong separation condition, then $E$ is totally disconnected. Notice that a totally disconnected Bedford-McMullen set is always normal and contains infinitely many components. Thus both the strong separation condition and the totally disconnected condition are stronger than the CSC.

\subsection{Bedford-McMullen sets with finitely many components}
In \cite{Xiao}, Xiao completely characterized fractal squares with finitely many components. Since  the  methods in \cite{Xiao} are topological, the results are applicable to Bedford-McMullen sets. In particular, Xiao proved that if a fractal square $F$ contains only finitely many components, then every component of $F$ intersects the boundary of $Q_0$, see \cite[Proposition~3.3]{Xiao}. Applying this result to Bedford-McMullen sets, and using Lemma~\ref{Thm}, we immediately have the following interesting result.
\begin{prop}
  Let $E$ be a normal Bedford-McMullen set which is not one-sided. Then $E$ contains only finitely many components if and only if every component of $E$ intersects the boundary of $Q_0$.
\end{prop}

\section*{Acknowledgments}
The authors are  grateful to the referees for their valuable
comments. The proof of Propostion~\ref{prop1} is suggested by the referees. They also wish to thank Prof. Kenneth
Falconer and Dr. Jian-Ci Xiao for their helpful comments. The Example~\ref{ex:Lipschitz} are based on the discussions with Dr. Xiao.

\end{document}